\documentclass[a4paper,12pt]{article}
\usepackage{latexsym,amssymb,amsfonts,amsmath,amsthm}
\usepackage[dvips]{graphicx}
\newtheorem{theorem}{Theorem}
\newtheorem{lemma}{Lemma}

\newtheorem{proposition}{Proposition}
\newtheorem{remark}{Remark}
\newtheorem{definition}{Definition}

\newcommand{\bs}[1]{\boldsymbol{#1}}

\newcommand{\fp}{\widetilde{g}^{(+)}}
\newcommand{\fm}{\widetilde{g}^{(-)}}
\newcommand{\lp}{\widetilde{\lambda}^{(+)}}
\newcommand{\lm}{\widetilde{\lambda}^{(-)}}

\title{{\Large {\bf One-dimensional quantum walks via generating function and the CGMV method   
}
}}
\author{ 
{\small 
Norio Konno,$^{1}$ 
\footnote{konno@ynu.ac.jp 
}\quad 
Etsuo Segawa,$^{2}$ 
\footnote{e-segawa@m.tohoku.ac.jp 
}\quad
}\\ 
{\scriptsize $^{1}$ 
Department of Applied Mathematics, Faculty of Engineering, Yokohama National University
}\\
{\scriptsize 
Hodogaya, Yokohama 240-8501, Japan
} \\
{\scriptsize $^2$ 
Graduate School of Information Sciences, Tohoku University
}\\
{\scriptsize 
Aoba, Sendai, 980-8579, Japan
} \\
} 
\date{\empty }
\begin{document}
\maketitle

\par\noindent
\begin{small}
\par\noindent
{\bf Abstract}
We treat a quantum walk (QW) on the line whose quantum coin at each vertex tends to be the identity 
as the distance goes to infinity. 
We obtain a limit theorem that this QW exhibits localization with not an exponential but a ``power-law" decay around the origin 
and a ``strongly" ballistic spreading called bottom localization in this paper. 
This limit theorem implies the weak convergence with linear scaling 
whose density has two delta measures at $x=0$ (the origin) and $x=1$ (the bottom) without continuous parts. 


\end{small}

\setcounter{equation}{0}

\section{Introduction}
Quantum walks (QWs) are considered as quantum counterparts of random walks~\cite{Meyer}. 
Primitive forms of QWs on lines have already appeared as discrete time and space analogue of a relativistic motion of a free particle 
known as the Feynman checkerboard \cite{Feynman}, and a toy model to construct the quantum probability theory discussed in \cite{Gudder}. 
The QWs on the line treated in our paper are denoted by a complex-valued sequence $(\gamma_j)_{j\in \mathbb{Z}}$  ($0\leq |\gamma_j|\leq 1$). 
The square absolute value of each parameter $\gamma_j$
assigned at each vertex $j$, $|\gamma_j|^2$, corresponds to a reflection strength at each vertex. 
For an extreme case, $|\gamma_j|=0$ for all $j$, the walk has no reflection. 
Thus in this case, the scaled distance from the origin weakly converges as follows: 
\[ X_n/n \Rightarrow \delta_1(x)\;\;(n\to\infty),  \]
where $\Rightarrow$ means weak convergence. 
The generating function is one of the effective tools to get stochastic behaviors of QWs with $|\gamma_j|<1$, 
for example, localization and weak convergence~\cite{Pemantle,CHKS,CKS,KLS}. 
The two functions $g_j^{(\pm)}:\mathbb{C}\to \mathbb{C}$ ($j\in \mathbb{Z}$) 
determined by the following continued fraction relationship give an expression for the 
generating function~\cite{KLS}. 
\begin{align}\label{intro}
	\fp_{j+1}(z) = -\frac{z^2\overline{\gamma}_{j+1}+\fp_{j}}{z^2+\gamma_{j+1}\fp_j}\;\;\mathrm{and}\;\;
        \fm_{j-1}(z) = -\frac{-z^2\gamma_{j-1}+\fm_j}{z^2-\overline{\gamma}_{j-1}\fm_j}. 
\end{align}
Indeed, for spatial homogeneous case, i.e., its coin parameters are $\gamma_j=\gamma$, 
the generating function plays an important role
to show that the walks belong to a universality class of QWs \cite{CHKS,CKS,KLS}, that is, 
\[ X_n/n \Rightarrow c\delta_0(x)+(1-c)r(x)f_K(x;\rho), \;\;(n\to\infty) \]
where $f_K(x;\rho)$ has anti-bell shape with the finite support $|x|\in [0,\rho)$ ($\rho=\sqrt{1-|\gamma|^2}$), 
and $c\in [0,1)$ reflects a localization property, and the rational polynomial $r(x)$ depends on the setting of the walk, 
for example, the initial state~\cite{CHKS,Konno}, boundary condition~\cite{CKS,LiuPeturante}, and spatial one defect~\cite{KLS}. 
Here, $f_K(x;\rho)$ was first introduced by \cite{Konno,Konno2} (2002)
\begin{equation}\label{fK} f_K(x;\rho)=\frac{\boldsymbol{1}_{\{|x|<\rho\}}(x)|\gamma|}{\pi (1-x^2)\sqrt{\rho^2-x^2}}. \end{equation}

It has been a quite interesting problem to get a limit theorem corresponding to the above weak convergences 
in the case of varied coin parameters in a whole vertices on the line~\cite{Joye,Hanover,GV}. 
Recently, it was shown that QWs on the line are described by the CMV matrix~\cite{CGMV0}. 
The CMV matrix is a minimal representation of recurrence relations 
between the orthogonal polynomials of a measure $\mu$ on the unit circle determined by the Schur parameter $(\gamma_0,\gamma_1,\dots)$~\cite{CMV1,CMV2}. 
There is a nice review on the relationship between QWs and the CMV matrix~\cite{CGMV}.
A well developed theory of the spectral analysis of the CMV matrix\cite{BS} advantages the analysis of QWs, 
especially, its spectrum. 
Indeed, for a sequence of homogeneous Schur parameter $\gamma_j=\gamma$ ($0<|\gamma|<1$), the spectral measure is described by 
\begin{equation}
	d\mu(e^{i\theta})=w(e^{i\theta})\frac{d\theta}{2\pi}+\sum_{j}c_j\delta(e^{i\theta}-e^{i\theta_j}), 
\end{equation}
where $w(e^{i\theta})$ is an absolutely continuous part which has a finite support 
$\{\theta\in[0,2\pi); |\cos\theta|\leq \rho\}$, and $c_j>0$ is a mass at $e^{i\theta_j}$. 
For the extreme cases, $|\gamma|=0$ implies $d\mu$ becomes the uniform measure on $[0,2\pi)$. 
The existence of the mass point gives localization property of the QW~\cite{KS}. 
On the other hand, the absolutely continuous part is related to a ballistic transport, 
in fact, the real part of the support for the spectral measure, $\rho$, reflects 
the stochastic property in the weak limit measure of the QW, that is, the walk frequency exists between $|x|\leq \rho n$. 
The Schur function $f:\mathbb{C}\to \mathbb{C}$ has an important role to describe the spectral measure. 
The Schur function is obtained by the following continuum fractional expression known as the Schur algorithm: 
\begin{equation}\label{SchurFunc_intro}
f_0(z) = f(z);\; f_{j+1}(z) = \frac{1}{z}\frac{f_j (z)-\gamma_j}{1-\overline{\gamma_j}f_j(z)}, \;j\geq  0.
\end{equation}

Comparing Eq.~(\ref{SchurFunc_intro}) with Eq.~(\ref{intro}), 
we get a relationship (see Lemma~\ref{bridgeGS}) which connects 
the spectral analysis of the CMV matrix between the generating function of QWs. 
The QW treated here has position dependent coin parameters which tends to 
zero as the distance from the origin goes to infinity. 
More concretely, the coin parameters are $\gamma_j=1/(r+j)$ for $j\in \{0,1,\dots\}$ with $r>1$.
Our results on the QWs treated here suggest that 
this relationship can be a key to get detailed stochastic behaviors of QWs from its spectral analysis. 
On the other hand, Refs.~\cite{GVWW} and \cite{BGVW} focus on an inclusive relationship between 
the recurrence properties proposed by \cite{SKJ,GVWW} and its spectrum. 
We obtain an explicit expression for the limit distribution outside the universality class. 
For our best knowledge, this is the first result on an explicit expression for the limit distribution of a position dependent QW over the whole vertices: 
we show co-existence of localization with power-law decay around the origin and a strongly ballistic transport. 
(See Theorems~\ref{localization} and \ref{weakconv}.)

This paper is organized as follows. 
In Sect.~2, we give the definitions of the QWs on the line and 
present its generating functions in a general setting. 
Connections between the CMV matrix and the QWs are devoted in Sect.~3. 
Using this relationship obtained by Sect.~3, 
we compute limit theorems concretely for a spatial dependent QW with the coin parameter $\gamma_j=1/(r+j)$ in Sect.~4. 
Finally, we give a summary and discussion in Sect~5.  
\section{Generating function of quantum walks}
From now on, we consider the three types of QWs; 
(i) first kind of QW on infinite half line (H-QW(1)), 
(ii) second kind of QW on half infinite line (H-QW(2)), and (iii) QW on doubly infinite line (D-QW). 
The detailed definitions are in the following:
\begin{definition}
We assign two dimensional unitary matrices $\{H_j\}_{j\in \mathbb{Z}}$ at each vertex of $\mathbb{Z}$.
Put the two Hilbert spaces considering here as 
$\mathcal{H}_+=\mathrm{span}\{\bs{\delta}_{j,L},\bs{\delta}_{j,R}; j\in \mathbb{Z}_+\}$ and 
$\mathcal{H}=\mathrm{span}\{\bs{\delta}_{j,L},\bs{\delta}_{j,R}; j\in \mathbb{Z}\}$. 
The walks start with the initial state $\alpha\bs{\delta}_{0,L}+\beta\bs{\delta}_{0,R}$ where $\alpha,\beta\in \mathbb{C}$ 
with $|\alpha|^2+|\beta|^2=1$. 
\begin{enumerate}
\item First kind of QW on half line (H-QW(1)) \\
Total state space: $\mathcal{H}_+$ \\
Time evolution $E=SC$ on $\mathcal{H}_+$: 
\begin{align}
C &= \bigoplus_{j\in \mathbb{Z}_+}H_j, \\
S\bs{\delta}_{j,L} &= 
	\begin{cases} \bs{\delta}_{j-1,L} & \text{: $j\geq 1$} \\ \bs{\delta}_{0,R} & \text{: $j=0$ }  \end{cases},\;
S\bs{\delta}_{j,R} = \bs{\delta}_{j+1,R}. 
\end{align}
\item Second kind of QW on half line (H-QW(2))\\
Total state space: $\mathcal{H}_+$ \\
Time evolution $E=SC$ on $\mathcal{H}_+$: 
We choose quantum coin at the origin and the initial coin state so that 
$\langle\bs{\delta}_{0,R},H_0\bs{\delta}_{0,R}\rangle=\langle\bs{\delta}_{0,L},H_0\bs{\delta}_{0,L}\rangle=0$, and $\beta=0$. 
Then 
\begin{align}
C &= \bigoplus_{j\in \mathbb{Z}_+}H_j 
	\mathrm{\;with\;} (H_0)_{\bs{\delta}_{0,J},\bs{\delta}_{0,J}}=0 \mathrm{\;for\;} J\in\{L,R\}. \\
S\bs{\delta}_{j,L} &= \bs{\delta}_{j-1,L},\;
S\bs{\delta}_{j,R} = \bs{\delta}_{j+1,R}. 
\end{align}
\item QW on doubly infinite line (D-QW) \\
Total state space: $\mathcal{H}$ \\
Time evolution $E=SC$ on $\mathcal{H}$: 
\begin{align}
C &= \bigoplus_{j\in \mathbb{Z}_+}H_j \\
S\bs{\delta}_{j,L} &= \bs{\delta}_{j-1,L},\;
S\bs{\delta}_{j,R} = \bs{\delta}_{j+1,R}. 
\end{align}
\end{enumerate}
\end{definition}
We should remark that the subspace of $\mathcal{H}_+$, $\mathrm{span}\{\{\bs{\delta}_{0,L}\}\cup\{\bs{\delta}_{j,L},\bs{\delta}_{j,R}; j\geq 1\}\}$, 
is invariant under the action of the time evolution of H-QW(2). 
Put 
$\langle \bs{\delta}_{j,L}, H_j\bs{\delta}_{j,L}\rangle=a_j$, 
$\langle \bs{\delta}_{j,L}, H_j\bs{\delta}_{j,R}\rangle=b_j$, 
$\langle \bs{\delta}_{j,R}, H_j\bs{\delta}_{j,L}\rangle=c_j$, 
$\langle \bs{\delta}_{j,R}, H_j\bs{\delta}_{j,R}\rangle=d_j$. 
The matrix representations for the time evolutions of H-QW(1), H-QW(2) and D-QW are expressed 
by $E_1$, $E_2$, and $E_D$, respectively as follows: 
\begin{equation*}
E_{1}=
\left[
\begin{array}{cc|cc|cc|cc|c}
0   & 0   & a_1 & b_1 &     &     &     &     &  \\
a_0 & b_0 & 0   & 0   &     &     &     &     &  \\ \hline
0   & 0   &     &     & a_2 & b_2 &     &     &  \\
c_0 & d_0 &     &     & 0   & 0   &     &     &  \\ \hline
    &     & 0   & 0   &     &     & a_3 & b_3 &  \\
    &     & c_1 & d_1 &     &     & 0   & 0   &  \\ \hline
    &     &     &     & 0   & 0   &     &     & \ddots \\
    &     &     &     & c_2 & d_2 &     &     &  \\ \hline
    &     &     &     &     &     & \ddots    &     & \ddots
\end{array}
\right], 
\;\;
E_{2}=
\left[
\begin{array}{c|cc|cc|cc|c}
0   &  a_1 & b_1 &     &     &     &     &  \\ \hline
0   &      &     & a_2 & b_2 &     &     &  \\
c_0 &      &     & 0   & 0   &     &     &  \\ \hline
    &  0   & 0   &     &     & a_3 & b_3 &  \\
    &  c_1 & d_1 &     &     & 0   & 0   &  \\ \hline
    &      &     & 0   & 0   &     &     & \ddots \\
    &      &     & c_2 & d_2 &     &     &  \\ \hline
    &      &     &     &     & \ddots    &     & \ddots
\end{array}
\right]
\end{equation*}

\begin{equation*}
E_{D}=
\left[
\begin{array}{c|cc|cc|cc|cc|c}
\ddots &        & \ddots &        &        &     &     &     &     &  \\ \hline 
       &        &        & a_{-1} & b_{-1} &     &     &     &     &  \\
\ddots &        &        & 0      & 0      &     &     &     &     &  \\ \hline
       & 0      & 0      &        &        & a_0 & b_0 &     &     &  \\
       & c_{-2} & d_{-2} &        &        & 0   & 0   &     &     &  \\ \hline
       &        &        & 0      & 0      &     &     & a_1 & b_1 &  \\
       &        &        & c_{-1} & d_{-1} &     &     & 0   & 0   &  \\ \hline
       &        &        &        &        & 0   & 0   &     &     & \ddots \\
       &        &        &        &        & c_0 & d_0 &     &     &  \\ \hline
       &        &        &        &        &     &     & \ddots    &     & \ddots
\end{array}
\right],
\end{equation*}
where the orders of the basis in the above matrices for H-QW(1), H-QW(2) and D-QW are 
$$((0,L),(0,R),(1,L),(1,R),(2,L),(2,R),\dots),\;((0,L),(1,L),(1,R),(2,L),(2,R),\dots)$$ $$\mathrm{and}$$ $$(\dots,(-1,L),(-1,R),(0,L),(0,R),(1,L),(1,R),\dots),$$ 
respectively. 
The following lemma obtained by \cite{CGMV, KS} is useful to simplify our model. 
\begin{lemma}
Put the quantum coin assigned at position $j$ described by complex valued parameter $\gamma_j$ with $|\gamma_j|<1$ as 
\begin{equation}
H_j^{(\gamma_j)}= \begin{bmatrix} \rho_j & \overline{\gamma}_j \\ -\gamma_j & \rho_j \end{bmatrix}
\end{equation} 
where $\rho_j=\sqrt{1-|\gamma_j|^2}$, 
and we put $\bs{\delta}_{j,L}\cong {}^T[1,0]$, $\bs{\delta}_{j,R}\cong {}^T[0,1]$ 
on the subspace of $\mathcal{H}_j\equiv \mathrm{span}\{\bs{\delta}_{j,L},\bs{\delta}_{j,R}\}$.
Then for each case of QW, i.e., H-QW(1), H-QW(2) and D-QW, with quantum coins $\{H_j\}_{j}$, 
there exists an infinite diagonal matrix $D$ on $\mathcal{H}$ and a sequence of complex-valued parameters $\bs{\gamma}\equiv (\gamma_j)_j$ with $|\gamma_j|<1$, such that 
\begin{equation}
E=D^\dagger U^{(\bs{\gamma})} D. 
\end{equation}
Here $U^{(\bs{\gamma})}=SC^{(\bs{\gamma})}$, where 
$C^{(\bs{\gamma})}=\bigoplus_j H_j^{(\gamma_j)}$. 
\end{lemma}
So we get a one-to-one correspondence between each QW and pair of the diagonal matrix $D$ and the parameters $(\gamma_j)_j$. 
From now on, we concentrate on the walks with the time evolution $U^{(\bs{\gamma})}$ for simplicity. 
We call $(\gamma_j)_j$ coin parameter. 

Let $\Xi_n: \mathbb{Z}\to M_2(\mathbb{C})$ be defined by 
\begin{equation}
\Xi_n(j) \equiv 
	\begin{bmatrix}
        \langle \bs{\delta}_{j,L}, {U^{(\bs{\gamma)}}}^n \bs{\delta}_{0,L}\rangle & \langle \bs{\delta}_{j,L}, {U^{(\bs{\gamma)}}}^n \bs{\delta}_{0,R}\rangle\\
        \langle \bs{\delta}_{j,R}, {U^{(\bs{\gamma)}}}^n \bs{\delta}_{0,L}\rangle & \langle \bs{\delta}_{j,R}, {U^{(\bs{\gamma)}}}^n \bs{\delta}_{0,L}\rangle
        \end{bmatrix}
\end{equation}
We put $\Xi_n(j)=0$ for negative $j$ in the cases of H-QW(1) and (2). 
Define 
\begin{equation}
P_j=\begin{bmatrix} \rho_j & \overline{\gamma}_j \\ 0 & 0 \end{bmatrix},\;
Q_j=\begin{bmatrix} 0 & 0 \\ -\gamma_j & \rho_j \end{bmatrix},\;
R_j=\begin{bmatrix} -\gamma_j & \rho_j \\ 0 & 0 \end{bmatrix},\;
S_j=\begin{bmatrix} 0 & 0 \\ \rho_j & \overline{\gamma}_j \end{bmatrix}.
\end{equation}
We call $\Xi_n$ weight of passages with length $n$ in the following sense:  
from a simple observation, we obtain the following recursion equations :
for H-QW(II) and D-QW, 
\begin{equation}
\Xi_n(j) = P_{j+1}\Xi_{n-1}(j+1)+Q_{j-1}\Xi_{n-1}(j-1), 
\end{equation}
and for H-QW(I)
\begin{equation}\label{weightpath}
\Xi_n(j) = 
	\begin{cases}
        P_{j+1}\Xi_{n-1}(j+1)+Q_{j-1}\Xi_{n-1}(j-1) & \text{$j\geq 1$,}\\
        P_{1}\Xi_{n-1}(j+1)+S_{0}\Xi_{n-1}(0) & \text{$j=0$.}
        \end{cases}
\end{equation}

For $z\in \mathbb{C}$ with $|z|\leq 1$, we denote a generating function with respect to time $n$ as 
$\widetilde{\Xi}_j(z)\equiv \sum_{n\geq 0}\Xi_n(j)z^n$. 
Here to express the generating function, 
we define $\widetilde{g}^{(\pm)}_j(z)$ in the following continued-fraction representation: for $j\in \mathbb{Z}$, 
\begin{align} 
	\fp_j(z) &= 
        \begin{cases}
        -\frac{z^2}{\gamma_{j+1}}\left( 1-\frac{\rho_{j+1}^2}{1+\gamma_{j+1}\fp_{j+1}(z)} \right), & \text{: $\gamma_{j+1} \neq 0$} \\
        z^2 \fp_{j+1} & \text{: $\gamma_{j+1}=0$} \label{GFunc+} 
        \end{cases}\\
        \fm_j(z) &= 
        \begin{cases}
        -\frac{z^2}{\overline{\gamma}_{j-1}}\left( 1-\frac{\rho_{j-1}^2}{1-\overline{\gamma}_{j-1}\fm_{j-1}(z)} \right), & \text{: $\gamma_{j-1}\neq 0$, } \\
        z^2\fm_{j-1}(z), & \text{: $\gamma_{j-1}=0. $}\label{GFunc-} 
        \end{cases}
\end{align} 
We call $\fp_j(z)$ (resp.~$\fm_j(z)$) ``positive (resp.~negative) $j$-th g-function" , respectively. 
We should note that for $j\geq 0$, $\fp_j(z)$ only depends on at most parameters $(\gamma_1,\gamma_2,\dots)$ and also
$\fp_j(z)$ only depends on at most parameters $(\gamma_{-1},\gamma_{-2},\dots)$. 
To emphasize its dependences, we sometimes denote
$\fp_j(z) \equiv \widetilde{g}_{j}^{(+),\; (\gamma_1,\gamma_2,\dots)}(z)$ for $j\geq 0$, and 
$\fm_j(z) \equiv \widetilde{g}_{j}^{(-),\; (\gamma_{-1},\gamma_{-2},\dots)}(z)$ for $j\leq 0$. 

Denote $F^{(+)}_n(j)$ (resp. $F^{(-)}_n(j)$) as the weight of all passages which start
from $j$ and return to the same position $j$ at time $n$ avoiding $\{i \in \mathbb{Z} : i<j\}$ (resp.
$\{i \in \mathbb{Z} : i>j\}$) throughout the time interval $0 < s < n$, respectively. 
Indeed, the generating function
$\widetilde{F}^{(\pm)}_j (z) \equiv \sum_{n\geq 1} F^{(\pm)}_j(n)z^n$ can be expressed by using $\fp_j(z)$ and $\fm_j(z)$ 
as follows: 
\begin{align}
\widetilde{F}^{(+)}_j(z) 
	&= \begin{cases}
        \fp_j(z)R_j & \text{: H-QW(2) and D-QW,}\\
        \fp_j(z)R_j+z\delta_0(j) S_0 & \text{: H-QW(1), } 
            \end{cases}   \\ 
\widetilde{F}^{(-)}_j(z) 
	&= \fm_j(z)S_j . 
\end{align}
Then we give an expression for the generating function using $\{\fp_j(z)\}_{j\in \mathbb{N}}$ and $\{\fm_j(z)\}_{j\in \mathbb{N}}$ 
in the following lemma. 
We put $\rho_{-1}=0$ for H-QW(1) and H-QW(2) cases. 
\begin{lemma} \label{acco} \noindent
 	\begin{enumerate}
	\item If $j=0$, then 
	\begin{equation}\label{pero0}
		\widetilde{\Xi}_0(z) 
		= \begin{cases}
                	\frac{1}{1-\overline{\gamma}_0z +(\gamma_0-z)\fp_0(z)}
                        \begin{bmatrix} 1-\overline{\gamma}_0z & \rho_0\fp_0(z) \\ \rho_0 z & 1+\gamma_0\fp_0(z) \end{bmatrix}, & \text{: H-QW(1), } \\
                        
                        \frac{1}{1-\fp_0(z)}
                        \begin{bmatrix} 1 & 0 \\ 0 & 1-\fp_0(z)\end{bmatrix}, & \text{: H-QW(2),} \\
                        
                	\frac{1}{1+\gamma_0\fp_0(z)-\overline{\gamma}_0\fm_0(z)-\fp_0(z)\fm_0(z)}
        		\begin{bmatrix} 1-\overline{\gamma}_0\fm_0(z) & \rho_0\fp_0(z) \\ \rho_0\fm_0(z) & 1+\gamma_0\fp_0(z)\end{bmatrix}, & \text{: D-QW}
                  \end{cases}      
	\end{equation}
        \item if $|j|\geq 1$, then
        \begin{equation} \label{pooh}\\                
		\widetilde{\Xi}_j(z)
		= \begin{cases}
                	\left\{\delta_1(j)+(1-\delta_1(j))\lp_{j-1}(z)\cdots\lp_{1}(z)\right\}
                 		\begin{bmatrix} \lp_{j}(z)\fp_j(z) \\ z\end{bmatrix}
                        	\begin{bmatrix} -\gamma_0 & \rho_0 \end{bmatrix}\widetilde{\Xi}_0(z) & \text{: $j\geq 1$, } \\
                                \\
                 	\left\{\delta_{-1}(j)+(1-\delta_{-1}(j))\lm_{j+1}(z)\cdots\lm_{1}(z)\right\}
                 		\begin{bmatrix} z \\ \lm_{j}(z)\fm_j(z) \end{bmatrix}
                        	\begin{bmatrix} \rho_0 & \overline{\gamma}_0 \end{bmatrix}\widetilde{\Xi}_0(z) & \text{: $j\leq -1$, }       
        	  \end{cases}    
         \end{equation}
         \end{enumerate} 
where $\lp_j(z)=z\rho_j/(1+\gamma_j\fp_j(z))$, $\lm_j(z)=z\rho_j/(1-\overline{\gamma}_j\fm_j(z))$. 
\end{lemma}
\begin{proof}
As consequences of the D-QW case, we obtain the H-QWs (1) and (2) cases as follows: 
We omit the proof of the D-QW because we can see the detailed proof in \cite{KLS}. 
\begin{enumerate}
\item H-QW(1) case:  
Put $\widetilde{\Xi}_{j,D}(z)$ 
as the generating function at position $j$ whose coin parameters are 
$\gamma_{2j}=\alpha_j$, $\gamma_{2j-1}= 0$, $(j\neq 0)$, $=-1$, $(j=0)$ for $j\in \mathbb{Z}$. 
We put its g-functions as $\{\widetilde{g}_{j,D}^{(+)}(z)\}_{j}$ and $\{\widetilde{g}_{j,D}^{(-)}(z)\}_{j}$. 
On the other hand, $\widetilde{\Xi}_{j,H(1)}(z)$ as the generating function at position $j$ 
whose coin parameters are $\gamma_{j}=\alpha_j$ for $j\in \mathbb{Z}_+$. 
We put its g-functions as $\{\widetilde{g}_{j,H(1)}^{(+)}(z)\}_{j}$ and $\{\widetilde{g}_{j,H(1)}^{(-)}(z)\}_{j}$. 
From the definition of H-QW(1), we have 
\begin{equation}\label{HD}
	\widetilde{\Xi}_{j,H(1)}(z^2)=\widetilde{\Xi}_{2j,D}(z),\;\;(j\in \mathbb{Z}_+)
\end{equation}
Note that 
\begin{equation}
	\widetilde{g}_{2j,D}^{(+)}(z)=-\frac{z^4}{\alpha_{j+1}}\left( 1-\frac{|\rho_{j+1}|^2}{1+\alpha_{j+1}\widetilde{g}_{2j+2,D}^{(+)}(z)} \right),
\end{equation}
which implies $\widetilde{g}_{2j,D}^{(+)}(z)=\widetilde{g}_{j,H(1)}^{(+)}(z^2)$. 
Also note that $\widetilde{g}_{0,D}^{(-)}(z)=z^2$. 
Therefore substituting the above expressions of $\widetilde{g}_{0,D}^{(\pm)}(z)$ into Eq.~(\ref{pero0}) for D-QW case and Eq.~(\ref{pooh}), 
we have the generating function for $\gamma_{2j}=\alpha_j$, $\gamma_{2j-1}= 0$, $(j\neq 0)$, $=-1$, $(j=0)$ ($j\in \mathbb{Z}$) case. 
Then from Eq.~(\ref{HD}), we obtain Eq.~(\ref{pero0}) for H-QW(1) case. 
\item H-QW(2) case: 
The definition of the H-QW(2) yields that replacing $\fm_0(z)=0$ and $\gamma_0=-1$ for D-QW case in EQ.~(\ref{pero0})
gives the generating function for H-QW(2) in Eq.~(\ref{pero0}). 
\end{enumerate}
\end{proof}
Combining Lemma \ref{acco} with the spatial Fourier analysis, we obtain the following weak limit theorems 
for H-QW(1)\cite{LiuPeturante}, H-QW(2)\cite{CKS}, and D-QW \cite{Konno} with homogeneous coin parameter 
$\gamma_j=\gamma\in \mathbb{C}$ with $0<|\gamma|<1$. 
\begin{theorem}[\cite{CKS},\cite{LiuPeturante},\cite{Konno}] \label{Weak}
The walks with coin parameter $\alpha_j=\gamma$ start at the origin with the initial coin states ${}^T[\alpha,\beta]$ (for H-QW(1), and D-QW), and 
${}^T[1,0]$ (H-QW(2)), respectively. 
Then we have 
\begin{equation}\label{WeakConv}
	\frac{X_n}{n} \Rightarrow c\delta_0(x) +w(x)f_K(x;\rho)\;\;(n\to\infty), 
\end{equation}
where 
\begin{align}
 c= &
        \begin{cases}
        \frac{\mathrm{Re}^2(\gamma)}{1-\mathrm{Im}^2(\gamma)}
        	\left| \alpha +\beta\nu(\gamma) \right|^2 \frac{1+\nu^2(\gamma)}{1-\nu^2(\gamma)} & \text{: H-QW(1)} \\
                \\
        \bs{1}_{\{|\gamma|^2+\mathrm{Re}(\gamma)>0\}}(\gamma) \frac{2(|\gamma|^2+\mathrm{Re}(\gamma))}{|1+\gamma|^2} & \text{: H-QW(2)} \\
        \\
        0 & \text{: D-QW}        
        \end{cases} \\
 w(x)= & 
 	\begin{cases}
        \bs{1}_{\{x\geq 0\}}(x)\frac{2|\gamma|^3/\rho^2\cdot \left\{ |\alpha|^2+|-\gamma\alpha+\rho\beta|^2 +2\rho \mathrm{Im}(\gamma)\mathrm{Im}(-\gamma\alpha+\rho\beta)\right\}x^2}
        	{\mathrm{Re}^2(\gamma)+\mathrm{Im}^2(\gamma)x^2} & \text{: H-QW(1)} \\
                \\
        \bs{1}_{\{x\geq 0\}}(x)\frac{2|\gamma|^2 (1+\mathrm{Re}(\gamma))x^2}{\left(|\gamma|^2+\mathrm{Re}(\gamma)\right)^2+\mathrm{Im}^2(\gamma)x^2} & \text{: H-QW(2)} \label{weight}\\
        \\
        1-\left(|\alpha|^2-|\beta|^2+2\mathrm{Re}(\gamma\alpha\overline{\beta})/\rho \right) x & \text{: D-QW}        
        \end{cases}       
\end{align}
Here $\rho=\sqrt{1-|\gamma|^2}$, and  $\nu(\gamma)=\mathrm{sgn}(\mathrm{Re}(\gamma))/\rho \cdot \left\{ \sqrt{1-\mathrm{Im}^2(\gamma)}-|\mathrm{Re}(\gamma)| \right\}$. 
\end{theorem}
The RHS of the first term in Eq.~(\ref{WeakConv}) provides localization (see Theorem.~\ref{homoloc}), and 
the second term provides the ballistic transport. 
The shape of the weight function $w(x)$ in Eq.~(\ref{weight}) depends on the boundary condition and initial state. 
On the other hand, we see $f_K(s;\rho)$ commonly for each case. 
In these cases, the $f_K(x;\rho)$ defined in Eq.~(\ref{fK}) appears 
as the Jacobian for the change of the variables $x=\partial \theta(k)/\partial k$, i.e., 
\begin{equation}  f_K(x;\rho)=\frac{1}{\pi\frac{\partial^2 \theta(k)}{\partial k^2}}\bigg|_{x=\partial \theta(k)/\partial k}, \end{equation}
where $\theta(k)$ is the argument of the singular point for the Fourier transform of the generating function. 
See Refs.~\cite{KLS,LiuPeturante}, for example, for more detailed computations aroud here.  
\section{Relation between spectral analysis of CMV matrix and generating function}
Let $L^2_\mu$ be the Hilbert space of $\mu$-square integrable functions whose inner product is 
\[ (f,g)_\mu=\int_{|z|=1}\overline{f(z)}g(z)d\mu,\;f,g\in L^2_\mu. \]
Let $U$ be a unitary time evolution on $\mathcal{H}$. 
We consider a complete orthogonal basis system of $L^2_\mu$ from the order-set $\{1, z^{-1},z, z^{-2},z^{2},\dots\}$.  
Put $\{\chi_j\}_{j}$ as the orthogonal basis system. 
Then $\chi_0(z) = 1$, 
\begin{align} 
\chi_{2j-1}(z) &\in \mathrm{span}\{1,z^{-1},z^1,\dots,z^{-j}\},\\
\chi_{2j}(z) &\in \mathrm{span}\{1,z^{-1},z^1,\dots,z^{-j},z^j\}.
\end{align}
Thus we obtain the following integral representation: 
\[ (\chi_m,z^n\chi_l)_\mu=\int_{|z|=1} z^n \overline{\chi_m(z)}\chi_l(z)d\mu(z) \]
In general, there exists a complex-valued sequence called Schur parameter $(\gamma_0,\gamma_1,\dots)$ with $|\gamma_j|<1$ which denotes five recursion relation 
between $\{\chi_m\}_m$, that is,  $(\mathcal{C})_{l,m}=(\chi_l,z\chi_m)_{\mu}$: 

\begin{equation*}
\mathcal{C}=
\begin{bmatrix}
\overline{\gamma}_0      & \rho_0\overline{\gamma}_1    & \rho_0\rho_1  & 0               & 0             & 0               & 0            & 0      & \ldots \\
\rho_0         & -\gamma_0\overline{a}_1      & -\gamma_0\rho_1    & 0               & 0             & 0                & 0             & 0       & \ldots \\
0              & \rho_1\overline{\gamma}_2    & -a_1\overline{\gamma}_2 & \rho_2\overline{\gamma}_3 & \rho_2\rho_3  & 0                & 0             & 0       & \ldots \\  
0              & \rho_1\rho_2       & -\gamma_1\rho_2    & -\gamma_2\overline{\gamma}_3   & -\gamma_2\rho_3    & 0               & 0            & 0       & \ldots \\
0              & 0                  & 0             & \rho_3\overline{\gamma}_4 & -\gamma_3\overline{\gamma}_4 & \rho_4\overline{\gamma}_5 & \rho_4\rho_5 & 0       & \ldots \\
0               & 0                   & 0              & \rho_3\rho_4    & -\gamma_3\rho_4    & -\gamma_4\overline{\gamma}_5   & -\gamma_4\rho_5   & 0      & \ldots  \\
               &\vdots              &\vdots         &\vdots           & \vdots        & \vdots          & \ddots       &  &         
\end{bmatrix},
\end{equation*}
where $\rho_j=\sqrt{1-|\gamma_j|^2}$. $\mathcal{C}$ is called the CMV matrix. 
Conversely, if we get $(\gamma_0,\gamma_1,\dots)$, the measure on the unit circle is uniquely determined. 
In fact, the following procedure is a standard method to get the measure from the Schur parameter. 
Let $f(z)$ be the Schur function with the Schur parameter $(\gamma_0,\gamma_1,\dots)$. 
The Schur function is obtained by the following continuum fractional expression known as Schur algorithm: 
\begin{equation}\label{SchurFunc}
f_0(z) = f(z);\; f_{j+1}(z) = \frac{1}{z}\frac{f_j (z)-\gamma_j}{1-\overline{\gamma_j}f_j(z)}, \;j\geq  0.
\end{equation}
To emphasize the dependence of $\gamma_j$'s, we express $f_j$ as $f_j^{(\gamma_0,\gamma_1,\dots)}$. 
The Caratheodory function of the Schur parameter $(\gamma_0,\gamma_1,\cdots)$ are defined by 
\begin{equation}\label{Carathe}
F(z)=\frac{1+zf(z)}{1-zf(z)}.
\end{equation}
An equivalent expression for $F(z)$ is 
\[ F(z)=\int_{|z|=1}\frac{t+z}{t-z}d\mu(t).  \]
The measure $d\mu$ is decomposed into $d\mu(e^{i\theta})=w(e^{i\theta})\frac{d\theta}{2\pi}+d\mu_s(e^{i\theta})$, 
where $w(e^{i\theta})$ is absolutely continuous part called weight, and $\mu_s$ is singular part. 
The weight can be obtained by
\begin{equation}\label{ac}
w(e^{i\theta})=\lim_{r\uparrow  1} \mathrm{Re}(F(re^{i\theta}))
\end{equation}
The singular points are concentrated on 
$\{e^{i\theta}: \lim_{r\uparrow  1} \mathrm{Re}(F(r(e^{i\theta})))=\infty\}$, moreover 
the mass points are given by 
\begin{equation}\label{mass}
\mu(\{e^{i\theta}\})=\lim_{r\uparrow  1} \frac{1-r}{2}F(re^{i\theta}).
\end{equation}
Now we present an expression of the CMV matrix by using second kind of QW. 
Let $\mathcal{H}_e\equiv \mathrm{span}\{\{\bs{\delta}_{0,L}\}\cup\{\bs{\delta}_{2j,R},\bs{\delta}_{2j,L}; j\geq 1\}\}\subset \mathcal{H}_+$, 
$\mathcal{H}_o\equiv \mathrm{span}\{\bs{\delta}_{2j+1,R},\bs{\delta}_{2j+1,L}; j\geq 0\}\subset\mathcal{H}_+$.  
Then we obtain a relation between QW and the CMV matrix: 

\begin{proposition}
Let $U$ be the time evolution of second kind of QW on half line. 
\begin{align}
\mathcal{C}\cong U^2 |_{\mathcal{H}_e},\;{}^{T}\mathcal{C}\cong U^2 |_{\mathcal{H}_o}. 
\end{align}
\end{proposition}

\begin{proof}
The CMV matrix $\mathcal{C}$ is decomposed into $\mathcal{C}=\mathcal{L}\mathcal{M}$, where
\begin{align}
	\mathcal{L} = \mathrm{diag}(\Theta_0,\Theta_2,\Theta_4,\dots), \;
        \mathcal{M} = \mathrm{diag}(1,\Theta_3,\Theta_5,\dots), \;\; 
        \mathrm{with}\;\; \Theta_j = \begin{bmatrix} \overline{\gamma}_j & \rho_j \\ \rho_j & -{\gamma}_j \end{bmatrix}. 
\end{align}
We interpret $\mathcal{M}$ as $\mathcal{M}: \mathcal{H}_e\to \mathcal{H}_o$ 
by relabeling raws and columns of $\mathcal{M}$ as \\
$(1,R),(1,L),(3,R),(3,L),\dots$ and $(0,L),(2,R),(2,L),(4,R),\dots$, 
respectively. 
A simple observation gives $\mathcal{M}=U|_{\mathcal{H}_e}$. 
Conversely, relabeling raws and columns of $\mathcal{L}$ as \\ $(0,L),(2,R),(2,L),(4,R),\dots$ as $(1,R),(1,L),(3,R),(3,L),\dots$, respectively, 
we obtain $\mathcal{L}=U|_{\mathcal{H}_o}$. 
So we complete the proof. 
\end{proof}
The existence of mass points, i.e., ``eigenvalues" in the spectrum of the corresponding CMV matrix 
ensures the localization of QWs. 
In fact, the limit measure of the localization is obtained by the orthogonal projection onto the eigenspaces of the initial state. 
We give the following theorem with respect to localization for a space homogeneous case whose parameters are $\gamma_j=\gamma$. 
\begin{theorem}[\cite{KS}]\label{homoloc}
The walks start at the origin with the initial coin state ${}^T[\alpha,\beta]$ (for H-QW(1)) and ${}^T[1,0]$ (for H-QW(2)), 
respectively. 
\begin{equation*}
\lim_{t\to\infty}P(X_t^{(I)}=j)
=\frac{|\mathrm{Re}(\gamma)|^2}{1-\mathrm{Im}(\gamma)^2}
|\alpha+\beta\nu_I(\gamma)|^2(1+\nu_I^2(\gamma))\nu_I^{2j}(\gamma),
\end{equation*}
\begin{equation*}
P(X^{(II)}_t=j)
\sim \bs{1}_{\{|\gamma|^2+\mathrm{Re}(\gamma)>0\}}(\gamma) \frac{1+(-1)^{n+j}}{2}
\left(\frac{|\gamma|^2+\mathrm{Re}(\gamma)}{(1+\gamma)^2}\right)^2
\left(1+\frac{\bs{1}_{\{j\geq 1\}}(j)}{\nu_{II}^2(\gamma)} \right) \nu_{II}^{2j}(\gamma), 
\end{equation*}
where
\begin{align*}
\nu_I(\gamma) =\frac{\mathrm{sgn} (\mathrm{Re}(\gamma))}{\rho} \left\{ \sqrt{1-\mathrm{Im}(\gamma)^2}-|\mathrm{Re}(\gamma)| \right\} 
\;\mathrm{and}\;
\nu_{II}(\gamma) = \frac{\rho}{|1+\gamma|}. 
\end{align*}
\end{theorem}
We should remark that the summation over the positions is strictly smaller than one. 
The missing value is characterized by Theorem.~\ref{Weak}. 

In the following, we propose a useful connection 
between spectral analysis of the CMV matrix and generating function of QWs. 
The derivation of the above lemma is due to just a simple comparing between Eqs.~(\ref{GFunc+}) and (\ref{SchurFunc}). 
\begin{lemma}\label{bridgeGS}
\begin{align}
	\widetilde{g}_{j}^{(+),\; (\gamma_1,\gamma_2,\dots)}(z) &= z^2{f}_j^{(\overline{\gamma}_1,\overline{\gamma}_2,\dots)}(z^2) & \text{: $j\geq 0$,} \\
	\widetilde{g}_{j}^{(-),\; (\gamma_{-1},\gamma_{-2},\dots)}(z) &= z^2{f}_{|j|}^{(-\gamma_{-1},-\gamma_{-2},\dots)}(z^2) & \text{: $j\leq 0$,}
\end{align}
\end{lemma}

In the next section, we show an example that the connection works well 
to get a stochastic behaviors of a QW whose quantum coin at each vertex tends to be transparent 
as the distance goes to infinity. 

\section{Stochastic behaviors of spatial dependent quantum walks}\label{CMV_Gen}
Now in this section, we consider H-QW(I) with coin parameter 
\begin{equation}\label{powerlawpara}
	\gamma_j=1/(r+j)\;\; (r>1).
\end{equation} 
For each time $n\in \mathbb{N}$, define a map $\mu_n: \mathbb{Z}_+\to [0,1]$ as 
\[ \mu_n(j)=||\Pi_{j}U^{(\bs{\gamma})}(\alpha\bs{\delta}_{0,L}+\beta\bs{\delta}_{0,R})||^2=||\Xi_n(j)\varphi_0||^2,  \]
where $\Pi_j$ is orthogonal projection onto subspace spanned by $\bs{\delta}_{j,L}$ and $\bs{\delta}_{j,R}$, 
and $\varphi_0={}^T[\alpha,\beta]$. 
From the unitarity of the time evolution, $\sum_{j\geq 0}\mu_n(j)=1$. So we call $\mu_n$ distribution of QW at time $n$. 
We also define a dual distribution of $\mu_n$ as $\widetilde{\mu}_n(j)\equiv {\mu}_n(n-j)$. 
In this section, our interest focuses on the sequence of distributions $(\mu_n)_{n\in \mathbb{N}}$ and $(\widetilde{\mu}_n)_{n\in \mathbb{N}}$. 
Put for each time $n$, a random variable $X_n$ following $\mu_n$, that is, $P(X_n=j)=\mu_n(j)$. 
To get its generating function, first of all, we consider its spectrum as follows. 
The Schur function for the Schur parameter $(\gamma_0,\gamma_1,\dots)$ is known as (see for example \cite{CGMV})
\begin{equation}\label{SF}
f_j(z)=\frac{1}{r+j-(r-1+j)z},\;\;j\in\{0,1,2,\dots\}.
\end{equation}
On the other hand, combining Lemma \ref{bridgeGS} with Eq.~(\ref{SF}) implies 
\begin{equation}\label{gj}
	\widetilde{g}_j^{(+),(\gamma_1,\gamma_2,\dots)}(z)=\frac{z^2}{r+1+j-(r+j)z^2}. 
\end{equation}
Equations (\ref{Carathe}) and (\ref{SF}) imply 
\begin{equation}\label{F}
F(z)=\frac{(z+1)(z-r/(r-1))}{(z-1)(z+r/(r-1))}. 
\end{equation}
Then from Eqs.~(\ref{ac}), (\ref{mass}) and (\ref{F}), the spectral measure is 
\begin{equation}\label{measure}
	d\nu(e^{i\theta})=\frac{\cos^2 \frac{\theta}{2}}{\cos^2 \frac{\theta}{2}+\frac{1}{4r(r-1)}} \frac{d\theta}{2\pi}+\frac{1}{2r-1}\delta(e^{i\theta}-1). 
\end{equation}
Note that the coin parameter $\gamma=0$ means that its quantum coin is the identity, that is, perfect transmission. 
In our model, since $\gamma_j\to 0$ $(j\to\infty)$, quantum coin assigned at each vertex tends to the identity as the distance from the origin increases. 
The full unit circle $\theta\in[0,2\pi)$ of the support for absolutely continuous part of RHS in Eq.~(\ref{measure}) reflects its argument, 
since the support of the measure for the QW with coin parameter $(0,0,\dots)$ is also on the full unit circle $\theta\in [0,2\pi)$. 
Therefore, at the first glance, it seems that the walker gets farther and farther away from the origin. 
Nevertheless, the measure has also mass point at $\theta=0$. 
This fact predicts the opposite property, that is, ``localization" at the origin as we have already seen other QW models, for example \cite{CGMV,KS}. 
So what happens to this walk? 
The following theorem gives its answer. 

\begin{theorem}\label{localization}
Assume that H-QW(1) whose coin parameter is $(\gamma_0,\gamma_1,\dots)$ with $\gamma_j=1/(r+j)$ $(r>1)$ is 
located at the origin with the initial coin state $\varphi_0={}^T[\alpha, \beta]$ $(|\alpha|^2+|\beta|^2=1)$. 
Then we have 
\begin{enumerate}
\item (Localization around the origin) 
\begin{equation}\label{origin}
\lim_{n\to \infty} \mu_n(j)=\frac{r^2-1}{(r-1+j)(r+1+j)}\mu_\infty^{o},
\end{equation}
where 
\[\mu_\infty^{o}= \frac{2r^2}{(1+r)^2(1-2r)^2}\left|\alpha\sqrt{1-1/r}+\beta\sqrt{1+1/r}\right|^2. \]
\item (Localization around the bottom)
\begin{equation}\label{bottom}
\lim_{n\to \infty} \widetilde{\mu}_n(j)=
	\begin{cases}
        \frac{r}{r+1}\left| -\alpha/r+\beta\sqrt{1-1/r^2} \right|^2 & \text{: $j=0$, }\\
        \\
        \frac{r-1}{r}\left| \alpha\sqrt{1-1/r^2}+\beta/r \right|^2 & \text{: $j=1$, }\\
        \\
        \frac{1}{r(r+1)(r-1)^2}\left| \alpha\frac{r^2+r-1}{r-1}-\beta\sqrt{r^2-1} \right|^2\left(\frac{1-r}{r}\right)^{2j} & \text{: $j\geq 2$. }
        \end{cases}
\end{equation}
\end{enumerate}
Moreover $c_0+c_1=1$, with $c_0=\sum_{j\geq 0} \lim_{n\to\infty}\mu_n(j)$, $c_1=\sum_{j\geq 0} \lim_{n\to\infty}\widetilde{\mu}_n(j)$. 
\end{theorem} 
The contribution of the mass point in Eq.~(\ref{measure}) is localization of the QW around the origin with not exponential but {\it power-law} decay. 
On the other hand, the absolutely continuous part in Eq.~(\ref{measure}) gives a ``strongly" ballistic behavior of the QW, in that, as $n\to \infty$, 
	\[ \mu_n([\tau (n-j)])\to \begin{cases} 0 & \text{$(0<\tau<1)$} \\ \widetilde{\mu}_\infty(j)>0 & \text{$(\tau=1)$} \end{cases}\]
Moreover $\widetilde{\mu}_\infty(j)$ is exponential decay, like usual localization however around {\it bottom}. 
We call ``bottom localization" to this kind of ballistic transport. 
Appropriate choice of initial coin state gives opposite properties in this walk, i.e., localizations at the origin and bottom simultaneously. 

By the way, let us consider the assertion of Eq.~(\ref{bottom}) in $j=0$ and $j=1$ cases. 
We can compute directly the weights of two kinds of paths, ``right $\to$ right $\to$ right $\to\cdots$ right" and 
``stay $\to$ right $\to$ right $\to\cdots$ right", which corresponds to $j=0$ and $j=1$ cases, respectively. 
From Eq.~(\ref{weightpath}), 
\begin{align*}
	\Xi_{n}(n) &= P_{n-1}\cdots P_1P_0=\rho_{n}\cdots \rho_1P_0 \\
                   &= \sqrt{\frac{r}{r+1}\left( 1+\frac{1}{r+n} \right)}P_0 \\ 
                   &\;\;\;\to \sqrt{\frac{r}{r+1}}P_0\;\; (n\to\infty) \\
        \Xi_{n-1}(n) &= P_{n-2}\cdots P_0S_0=\rho_{n-1}\cdots \rho_0 S_0 \\
                   &= \sqrt{\frac{r-1}{r}\left( 1+\frac{1}{r+n-1} \right)}S_0 \\ 
                   &\;\;\;\to \sqrt{\frac{r-1}{r}}S_0\;\; (n\to\infty)           
\end{align*}
Indeed, we have 
\begin{align*}
	\lim_{n\to\infty}\widetilde{\mu}_n(n) &= \frac{r}{r+1}||P_0\varphi_0||^2=\frac{r}{r+1}\left| -\alpha/r+\beta\sqrt{1-1/r^2} \right|^2, \\
        \lim_{n\to\infty}\widetilde{\mu}_n(n-1) &= \frac{r-1}{r}||S_0\varphi_0||^2=\frac{r-1}{r}\left| \alpha\sqrt{1-1/r^2}+\beta/r \right|^2. 
\end{align*}
Thus 
\begin{remark}
Equation (\ref{bottom}) in cases $j=0$ and $j=1$ yields that the two infinite paths itself are ``localized" in the above sense. 
\end{remark}

\begin{figure}[htbp]
\begin{center}
	\includegraphics[width=100mm]{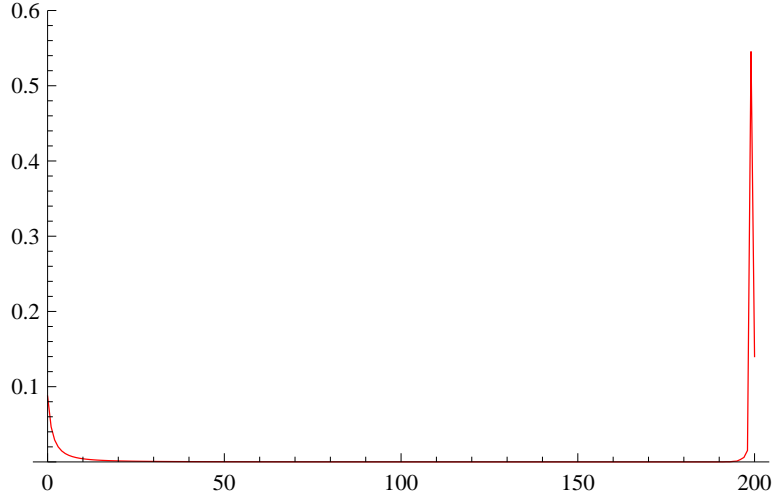}
\end{center}
\caption{Numerical simulation of H-QW(1) with coin parameter $\gamma_j=1/(3+j)$: 
the horizontal and vertical axes depict position and $\mu_{200}$ with the initial coin state $[1/\sqrt{2},1/\sqrt{2}]$, respectively. }
\label{fig:one}
\end{figure}

To prove Theorem \ref{localization}, we prepare the following key lemma 
with respect to the weight of passage $\Xi_j(n)$ in a closed form. 
We give its proof in the last part of this section.  
\begin{lemma}\label{decom}
The weight of passage is decomposed into three parts. 
\begin{equation}\label{Xi} \Xi_{j}(n)=L_{j}(n)+B_{j}(n)+I_{j}(n), \end{equation}
where $L_{j}(n)$, $B_{j}(n)$ and $I_{j}(n)$ are defined in Eqs.~(\ref{BLI}), (\ref{BLI1}) and (\ref{BLI2}), respectively. 
Moreover for fixed $j$, $\lim_{n\to\infty}I_{j}(n)=\lim_{n\to\infty}I_{n-j}(n) = 0$, 
\begin{align}
\lim_{n\to\infty}L_{j}(n) &= \sqrt{\frac{r(r-1)}{(r+1)(2r-1)^2}}\frac{1}{\sqrt{r+j}}\begin{bmatrix}1/\sqrt{r+1+j} \\ 1/\sqrt{r-1+j} \end{bmatrix}{}^\dagger\bs{l} \\
\lim_{n\to\infty}B_{n-j}(n) &= -\frac{1}{\sqrt{r(r+1)(r-1)^4}}\begin{bmatrix}1 \\ 0 \end{bmatrix} 
	\left\{ \left( \frac{1-r}{r} \right)^j {}^\dagger\bs{b}^{E}+\delta_0(j){}^\dagger\bs{b}^{0}+\delta_1(j){}^\dagger\bs{b}^{1} \right\} 
\end{align}
and $\lim_{n\to\infty}B_{j}(n)=\lim_{n\to\infty}L_{n-j}(n)=0$. \\
Here
\begin{align}
\bs{b}^E &= \begin{bmatrix} 1-r-r^2 \\ (r-1)\sqrt{r^2-1} \end{bmatrix}, \;\bs{b}^0 = \begin{bmatrix} 2r-1 \\ -(r-1)\sqrt{r^2-1} \end{bmatrix},  \;\bs{b}^1 = \begin{bmatrix} 1 \\ 0 \end{bmatrix},\;
\bs{l} = \begin{bmatrix} \sqrt{r-1} \\ \sqrt{r+1} \end{bmatrix}. 
\end{align}
\end{lemma}
\begin{remark}
From the linear independence of $\bs{b}^{0}$, $\bs{b}^{1}$ and $\bs{b}^{E}\in \mathbb{C}^2$, 
the strongly ballistic spreading, i.e., bottom localization, is always ensured for all initial coin state $\varphi_0$ in this walk. 
On the other hand, choosing an appropriate initial coin state so that $\langle \varphi_0,\bs{l}\rangle=0$ eliminates localization at the origin. 
\end{remark}
%
\noindent {\it Proof of Theorem~\ref{localization}. }
The first term of RHS in Eq.~(\ref{Xi}), $L_j(n)$, contributes localization with power-law around the origin 
corresponding to part (1) in Theorem~\ref{localization}, 
while the second term, $B_j(n)$, gives the bottom localization corresponding to part (2). 
The finial term, $I_j(n)$, is an intermediate term between first and second terms. 
Lemma~\ref{decom} implies that 
\begin{equation}\label{Eq1}
\mu_\infty(j)=|| \lim_{n\to\infty}L_j(n)\varphi_0 ||^2,\;\widetilde{\mu}_\infty(j)=|| \lim_{n\to\infty}B_{n-j}(n)\varphi_0 ||^2. 
\end{equation}
Algebraic computations of RHS in Eq.~(\ref{Eq1}) give desired conclusion. \begin{flushright}$\square$\end{flushright}
%
\begin{theorem}\label{weakconv}\noindent
\begin{enumerate}
\item (Weak convergence)
\begin{equation}
\frac{X_n}{n}\Rightarrow c_0\delta_0(x)+c_1\delta_1(x)\;\;\;(n\to \infty), 
\end{equation}
where $c_0=\sum_{j\geq 0}\mu_\infty(j)$, $c_1=1-c_0=\sum_{j\geq 0}\widetilde{\mu}_\infty(j)$. 
\item (Large deviation type convergence)
For the initial state $\varphi$ with $\langle \varphi_0, \bs{l}\rangle=0$, we have 
\begin{equation}
\lim_{n\to \infty}\frac{1}{n}\log P\left(1-\frac{X_n}{n}>\epsilon\right)=\epsilon \log \left(\frac{1-r}{r}\right)^2. 
\end{equation}
\end{enumerate}
\end{theorem}
\begin{proof}\noindent
\begin{enumerate}
\item 
As a consequence of Theorem \ref{localization}, we immediately obtain part (1).  
\item
Now we consider 
\begin{equation}\label{Eq2} 
P(1-X_n/n>\epsilon)=\sum_{j=[n\epsilon]}^{n}\langle  \Xi_{n-j}(n)\varphi_0,  \Xi_{n-j}(n)\varphi_0\rangle 
\end{equation}
for large $n$. 
The initial state $\langle\varphi_0,\bs{l}\rangle=0$ implies that we only estimate $B_{n-j}(n)$ and $I_{n-j}(n)$, that is, 
\begin{multline*} 
	\Xi_{n-j}^*(n)\Xi_{n-j}(n) = \\
		B_{n-j}^*(n)B_{n-j}(n)+I_{n-j}^*(n)I_{n-j}(n)+B_{n-j}^*(n)I_{n-j}(n)+I_{n-j}^*(n)B_{n-j}(n). 
\end{multline*}
We extract the essential parts of $B_{n-j}(n)^*B_{n-j}(n)$, $B_{n-j}(n)^*I_{n-j}(n)$ and $I_{n-j}(n)^*I_{n-j}(n)$ 
which are directly related to the summation in RHS of Eq.~(\ref{Eq2}) as follows:  
\[\left( 1-\frac{1}{r+n-j} \right)\tau^{2j},\;\;
\frac{\tau^{2j}}{r+n-j},\;\; \mathrm{and}\; \frac{\tau^{2j}}{(r+n-j)(r+n-j+1)}, \]
respectively, where $\tau=(1-r)/r$ (see Eqs.~(\ref{BLI}) and (\ref{BLI2}) for the explicit expressions for $B_j(n)$ and $I_{j}(n)$). 
Put 
\[ A_{n,\epsilon}\equiv \sum_{j= [n\epsilon]}^{n}\tau^{2j}\sim \frac{\tau^{2n\epsilon}}{1-\tau^2}(1-\tau^{2n(1-\epsilon)}). \]
Note that 
\[ \sum_{j\geq n\epsilon}^{n}\frac{\tau^{2j}}{(r+n-j)(r+n-j+1)}\leq \sum_{j\geq n\epsilon}^{n}\frac{\tau^{2j}}{r+n-j}\leq \frac{1}{r}A_{n,\epsilon}. \]
Then we have the RHS of Eq.~(\ref{Eq2}) is rewritten as 
\begin{equation}
\sum_{j=[n\epsilon]}^{n}\langle  \Xi_{n-j}(n)\varphi_0,  \Xi_{n-j}(n)\varphi_0\rangle \sim \frac{\tau^{2n\epsilon}}{1-\tau^2}\left(1+O(1)\right).
\end{equation}
Therefore
\[ \lim_{n\to\infty} \frac{1}{n}\log P(1-X_n/n>\epsilon)=\epsilon \log \tau^2.  \]
\end{enumerate}
\end{proof}
Finally, we give the proof of Lemma~\ref{decom}. \\
\noindent \\
\noindent {\it Proof of Lemma~\ref{decom}. } \\
At first, we decompose $\Xi_{n-j}(n)$ as $\Xi_{n-j}(n)=\Pi_{L}\Xi_{n-j}(n)+\Pi_{R}\Xi_{n-j}(n)$, 
where $\Pi_J$ is projection onto basis $\bs{\delta}_J$ ($J\in\{L,R\}$). 
For small $\delta<1$, 
\[ \Xi_{j}(n) = \frac{1}{2\pi i}\int_{|z|=\delta} \widetilde{\Xi}_j(z)\frac{dz}{z^{n+1}}. \]
So we have 
\begin{align}
	  \Pi_{L}\Xi_{n-j}(n) &= \frac{1}{2\pi i} \int_{|z|=\delta} \lp_{n-j}\cdots \lp_{1}\fp_{n-j}
          			\begin{bmatrix}-1/r+z & \sqrt{1-1/r^2} \\ 0 & 0 \end{bmatrix}\frac{1}{\widetilde{\Lambda}_0} \frac{dz}{z^{n+1}}, \label{EqA0}\\
          \Pi_{R}\Xi_{n-j}(n) &= \frac{1}{2\pi i} \int_{|z|=\delta} \lp_{n-j-1}\cdots \lp_{1}
          			\begin{bmatrix}0 & 0 \\ -1/r+z & \sqrt{1-1/r^2} \end{bmatrix}\frac{1}{\widetilde{\Lambda}_0} \frac{dz}{z^{n}}. \label{EqA0'}                      
\end{align}
Since $\lp_j$ is expressed by 
\[ \lp_j=z\sqrt{\frac{r+j-1}{r+j+1}}\frac{(r+j+1)-(r+j)z^2}{(r+j)-(r+j-1)z^2}, \]
we obtain 
\begin{align}\label{EqA1} 
	\lp_{j}\cdots \lp_{1}=\frac{1+1/r}{z^2-(r+1)/r} \left\{ \sqrt{\frac{r+j}{r+j-1}} -z^2 \sqrt{\frac{r+j-1}{r+j}} \right\}z^{j-1}. 
\end{align}
Moreover 
\begin{equation}\label{EqA2} 
\frac{1}{\widetilde{\Lambda}_0}= \frac{-r^2}{1-r^2}\frac{z^2-\frac{r+1}{r}}{\left(z-\frac{r}{1-r}\right)(z-1)}. 
\end{equation}
Substituting Eqs.~(\ref{EqA1}) and (\ref{EqA2}) into Eqs.~(\ref{EqA0}) and (\ref{EqA0'}), 
\begin{multline}\label{weight1}
	\Pi_{L}\Xi_{n-j}(n) = \frac{-\sqrt{\frac{r^3}{(r-1)^2(r+1)}}}{\sqrt{(r+n-j+1)(r+n-j)}} \\
        	\times \int_{|z|=\delta}\frac{1}{\left(z-\frac{r}{r+1}\right)(z-1)}
          			\begin{bmatrix}-1/r+z & \sqrt{1-1/r^2} \\ 0 & 0 \end{bmatrix} \frac{dz}{2\pi iz^{j-1}},
\end{multline}
\begin{multline}\label{weight2}                                
        \Pi_{R}\Xi_{n-j}(n) = \sqrt{\frac{r^3}{(r-1)^2(r+1)}}\sqrt{1-\frac{1}{r+n-j}} \\
        	\times \int_{|z|=\delta}\left\{\frac{z+1}{z-\frac{r}{1-r}}-\frac{1}{r+n-j-1}\frac{1}{\left(z-\frac{r}{1-r}\right)(z-1)}\right\} \\
          			\begin{bmatrix}0 & 0 \\ -1/r+z & \sqrt{1-1/r^2} \end{bmatrix} \frac{dz}{2\pi iz^{j+1}}.   
\end{multline}
Direct computations of the residues at $z=0$ in the integrands of RHSs of Eqs.~(\ref{weight1}) and (\ref{weight2}), respectively, lead to 
\[ \Xi_{j}(n)=B_j(n)+L_{j}(n)+I_j(n), \]
where
\begin{multline}\label{BLI}
B_{n-j}(n) = -\sqrt{1-\frac{1}{r+n-j}}\frac{1}{\sqrt{r(r+1)(r-1)^4}}\begin{bmatrix}1 \\ 0 \end{bmatrix} \\
	\times \left\{ \left( \frac{1-r}{r} \right)^j {}^\dagger\bs{b}^{E}+\delta_0(j){}^\dagger\bs{b}^{0}+\delta_1(j){}^\dagger\bs{b}^{1} \right\},
\end{multline}
\begin{align}
L_{j}(n) &= \sqrt{\frac{r(r-1)}{(r+1)(2r-1)^2}}\frac{1}{\sqrt{r+j}}\begin{bmatrix}1/\sqrt{r+1+j} \\ 1/\sqrt{r-1+j} \end{bmatrix}{}^\dagger\bs{l}, \label{BLI1}\\
I_j(n) &= \frac{\left(\frac{1-r}{r}\right)^{n-j}}{\sqrt{r+j}}\begin{bmatrix}r^3/((r-1)^4(2r-1)^2(r+1))\cdot 1/\sqrt{r+j+1} \\ 1/\sqrt{r(r+1)}\cdot 1/\sqrt{r+j-1} \end{bmatrix}{}^\dagger\bs{b}^E. \label{BLI2}
\end{align}
Obviously, from Eqs.~(\ref{BLI}), (\ref{BLI1}) and (\ref{BLI2}), 
we obtain $\lim_{n\to\infty}I_{j}(n)=\lim_{n\to\infty}I_{n-j}(n)=\lim_{n\to\infty}B_{j}(n)=\lim_{n\to\infty}L_{n-j}(n)=0$, and 
\begin{align*}
\lim_{n\to\infty}L_{j}(n) &= \sqrt{\frac{r(r-1)}{(r+1)(2r-1)^2}}\frac{1}{\sqrt{r+j}}\begin{bmatrix}1/\sqrt{r+1+j} \\ 1/\sqrt{r-1+j} \end{bmatrix}{}^\dagger\bs{l}, \\
\lim_{n\to\infty}B_{n-j}(n) &= -\frac{1}{\sqrt{r(r+1)(r-1)^4}}\begin{bmatrix}1 \\ 0 \end{bmatrix} 
	\left\{ \left( \frac{1-r}{r} \right)^j {}^\dagger\bs{b}^{E}+\delta_0(j){}^\dagger\bs{b}^{0}+\delta_1(j){}^\dagger\bs{b}^{1} \right\}. 
\end{align*}
We complete the proof. 
\begin{flushright}$\square$\end{flushright}
%

\section{Summary and Discussion}
We considered a connection between the Schur function which gives the spectrum of the CMV matrix 
and a generating function of QWs (Lemma~\ref{bridgeGS}). 
We presented an application of this relationship to analysis of stochastic behavior of QWs (Section~\ref{CMV_Gen}). 
In the H-QW(1) with parameters defined by Eq.~(\ref{powerlawpara}), we showed that opposite properties happen  simultaneously, that is, 
localization with a power-law decay around the origin and an ``extreme" ballistic spreading called bottom localization. 

Finally, we discuss the second kind of QW on the half line corresponding to the first kind of QW discussed in Sect.~\ref{CMV_Gen}. 
As we will see below, the non existence of just one self loop at the origin has large effect on a behavior of the bottom localization. 
Indeed, from a similar fashion of Sect.~\ref{CMV_Gen}, applying Lemmas \ref{acco} and \ref{bridgeGS} to H-QW(2), we have the 
following limit theorem corresponding to Theorem~\ref{localization}. 
\begin{theorem}\label{localization2}
Limit theorems for H-QW(2)
\begin{enumerate}
\item (Localization around the origin) 
\begin{equation}\label{origin2}
	\lim_{n\to \infty} \mu_n(j)\sim \frac{1+(-1)^{n+j}}{2}
        \times \begin{cases}
        	\left(\frac{1}{r+1}\right)^2 & \text{: $j=0$,}\\
                \frac{{2r}/(r+1)}{(r-1+j)(r+1+j)} & \text{: $j\geq 1$}
               \end{cases}
\end{equation}
\item (Localization around the bottom)
\begin{equation}\label{bottom2}
\lim_{n\to \infty} \widetilde{\mu}_n(j)= \left(1-\frac{1}{1+r}\right) \delta_0(j). 
\end{equation}
\end{enumerate}
Moreover $c_0+c_1=1$, with $c_0=\sum_{j\geq 0} \lim_{n\to\infty}\mu_n(j)$, $c_1=\sum_{j\geq 0} \lim_{n\to\infty}\widetilde{\mu}_n(j)$. 
\end{theorem}
Therefore, the contribution of the bottom localization for H-QW(2) is just nothing but 
the weight of ``right $\to$ right $\to\cdots$ " path itself. 
We leave the doubly infinite case to the readers applying Lemmas \ref{acco} and \ref{bridgeGS}. 
%


\par
\
\par\noindent
\noindent
{\bf Acknowledgments.}
NK acknowledges financial support of 
the Grant-in-Aid for Scientific Research (C) of Japan Society for the Promotion of Science (Grant No. 21540118). 
ES's work was partially supported by 
the Grant-in-Aid for Young Scientists (B) of Japan Society for the Promotion of Science (Grant No. 25800088).
\par
\
\par

\begin{small}
\bibliographystyle{jplain}

\end{small}


\end{document}